\documentclass[preprint,11pt]{elsarticle}

\usepackage[T1]{fontenc}
\usepackage{lmodern}
\usepackage{amsmath,amssymb,amsthm}
\usepackage{booktabs}
\usepackage{etoolbox}
\usepackage{needspace}
\usepackage{microtype}
\usepackage{xcolor}
\usepackage{tikz}
\usetikzlibrary{arrows.meta,calc,decorations.pathreplacing}
\usepackage{pgfplots}
\pgfplotsset{compat=1.18}
\usepackage{algorithm}
\usepackage[noend]{algpseudocode}
\usepackage[hidelinks]{hyperref}
\hypersetup{pdftitle={Computing the Nucleon of Matching Games},
            pdfauthor={Haris Aziz, Gergely Cs\'aji}}
\biboptions{sort&compress}
\makeatletter
\def\ps@pprintTitle{%
  \let\@oddhead\@empty
  \let\@evenhead\@empty
  \def\@oddfoot{\reset@font\hfil\thepage\hfil}%
  \let\@evenfoot\@oddfoot}
\makeatother

\theoremstyle{plain}
\newtheorem{theorem}{Theorem}
\newtheorem{proposition}{Proposition}
\newtheorem{lemma}{Lemma}
\theoremstyle{definition}
\newtheorem{definition}{Definition}
\newtheorem{example}{Example}
\newtheorem{remark}{Remark}

\AtBeginEnvironment{theorem}{\Needspace{7\baselineskip}}
\AtBeginEnvironment{lemma}{\Needspace{9\baselineskip}}
\AtBeginEnvironment{proposition}{\Needspace{14\baselineskip}}
\AtBeginEnvironment{definition}{\Needspace{9\baselineskip}}
\AtBeginEnvironment{example}{\Needspace{10\baselineskip}}
\AtBeginEnvironment{proof}{\Needspace{4\baselineskip}}
\AtBeginEnvironment{thebibliography}{\interlinepenalty=10000}

\newcommand{\R}{\mathbb{R}}
\newcommand{\rev}[1]{{\color{black}#1}}% revisions
\definecolor{changeTeal}{HTML}{008C8C}
\newcommand{\core}{\operatorname{core}}
\newcommand{\nucleon}{\operatorname{nucleon}}

\newcommand{\lin}{\operatorname{lin}}

\newcommand{\LSA}{\textnormal{\textsc{LSA-MinExcess}}}
\newcommand{\SNZC}{\textnormal{\textsc{Shortest Non-Zero Cycle}}}
\newcommand{\step}[1]{\smallskip\noindent\emph{#1}}

\definecolor{cblue}{HTML}{0072B2}
\definecolor{corange}{HTML}{D55E00}
\definecolor{cgreen}{HTML}{009E73}
\definecolor{cgray}{HTML}{6E6E6E}
\tikzset{
  player/.style={circle,draw,thick,fill=white,inner sep=0pt,
                 minimum size=6mm,font=\small},
  payoff/.style={font=\small,text=cblue},
  wt/.style={font=\footnotesize,fill=white,inner sep=1.2pt},
}

\begin{document}
\begin{frontmatter}

\title{Computing the Nucleon of Matching Games}
\author{Haris Aziz}
\affiliation{organization={UNSW Sydney}, city={Sydney}, country={Australia}}
\author{Gergely Csáji}
\affiliation{organization={HUN-REN Institute for Computer Science and Control}, city={Budapest}, country={Hungary}}

\begin{abstract}
The nucleon is the multiplicative counterpart of the nucleolus: it
lexicographically maximizes {the fraction of its stand-alone value
that each coalition receives}. 
We study exact nucleon computation in matching games and in two extensions:
$b$-matching games, in which players may have several partners, and
partitioned matching games, in which one player may own several vertices.
We resolve three open problems from a recent survey. For ordinary matching
games with nonnegative rational weights, we give a polynomial-time
combinatorial algorithm using maximum-weight matching, shortest paths, and
minimum-cost flow, ; it returns both a nucleon
allocation and a compact description of the full nucleon. In contrast,
computing any nucleon allocation is NP-hard for {unit-weight}
$b$-matching games with capacities at most three and for {unit-weight}
partitioned matching games of width at most three. Finally, for the unresolved capacity-two case, we show
that exact nucleon computation is polynomial-time Turing reducible to
\SNZC{}.
\end{abstract}


\end{frontmatter}

\section{Introduction}\label{sec:intro}

{Many cooperative markets face two intertwined decisions: which
pairwise collaborations to form, and how to divide the value they create.
A weighted graph captures the technological side of this problem: vertices
are participants, edges are feasible partnerships, and an edge weight is
the value generated by that pair. If each participant can join at most one
partnership, the efficient outcome is a maximum-weight matching. The same
optimization also determines bargaining power, because any coalition can
walk away and realize the value of a maximum-weight matching using only its
own members. This gives a matching game, a compact cooperative model in
which combinatorial efficiency and coalitional fairness are directly
linked. Allowing a participant to enter several partnerships gives a
$b$-matching game; allowing one economic player to own several vertices,
as when a country contributes several patient--donor pairs to an
international kidney exchange, gives a partitioned matching game
\citep{BBK+25a}.}

{{These games arise in several applications.} In pairwise kidney exchange, incompatible patient--donor pairs swap donors
two at a time, so the number of transplants that a set of pairs can
achieve is \rev{twice} the size of a maximum matching in their compatibility graph
\citep{RSU05a}. Because larger pools yield more transplants, national
programmes have an incentive to merge their pools, and the question
becomes how to allocate transplants among countries so that each
country's benefit reflects its contribution \citep{KVPS21a,BBK+25a}.
Recent work uses cooperative solution concepts, including the Shapley
value and the nucleolus, to define each country's fair share of
transplants \citep{BBPY24a}. In
collaborative transportation, firms combine their operations to reduce
costs, and allocation rules from cooperative game theory are standard
tools for sharing the savings \citep{GuRo16a}; when the savings arise
from pairing two partners, for instance by combining one carrier's
delivery route with another's return trip, the resulting game is a
matching game. Matching games also formalize roommate and partnership
formation with transferable utility \citep{EK01a,TaYa11a}, and
$b$-matching games model settings in which each agent can take part in
several pairings, such as arranging fixtures between sports teams or
players \citep{IS07a,BKPW18a}. %In each case, the total value is a single maximum-weight matching, while the bargaining power of each group of players comes from the matchings it could form on its own.
}

The core asks that every coalition receive at least what it can earn on
its own. This requirement expresses coalitional stability, but it may
be impossible to meet. Consider three firms, each able to partner with
either of the other two, with every partnership worth one. Their joint
value is one, while each pair can also earn one. Paying all three pairs
at least one would require the total allocation to be at least $3/2$.
An allocation concept that remains defined when the core is empty is
therefore useful even in the smallest matching markets.

The nucleon, introduced by \citet{FKFH98a}, compares each coalition's
payoff with its value through a ratio. It first makes the smallest ratio
as large as possible, then improves the next-smallest ratio, and
continues lexicographically. This evaluates shortfalls relative to the
opportunities of the affected coalition. For example, a coalition paid
$90$ when it can earn $100$ receives a larger fraction of its value than a
coalition paid $1/2$ when it can earn $1$. The nucleon gives priority to
the latter shortfall. The nucleolus instead compares additive excesses,
so the former shortfall is larger under its objective. Both concepts
refine the idea of protecting the least satisfied coalitions, but they
measure satisfaction differently.

{Ratios are natural when coalitions differ widely in size. A guarantee
such as ``every group of countries receives at least $90\%$ of the
transplants it could organize on its own'' is meaningful to every
participant, whereas an additive guarantee of ``at most two transplants
fewer'' is negligible for a large coalition and decisive for a small one.
\rev{The first stage of the nucleon maximizes the common ratio guaranteed
to proper coalitions. Capping this ratio at one gives the best
multiplicative approximate-core guarantee for all coalitions, a central
notion in cost sharing and in approximation algorithms for cooperative
games \citep{FaKe93a,JaMa07a,Vazi22a}.} For matching games this guarantee is
always at least $2/3$, and the triangle above shows that $2/3$ cannot be
improved \citep{FaKe93a,Vazi22a}. Whenever the core is nonempty, the
nucleon lies in the core. The later stages refine the guarantee
lexicographically, so that no coalition receives a smaller fraction of its
value than necessary. The nucleon thus combines the worst-case guarantee
of approximation algorithms with the lexicographic fairness of the
nucleolus, which makes it a natural allocation rule for the applications
above.}

Can the nucleon be computed directly from the graph, and how does the
answer change when players may have several partners or own several
vertices? \citet{FKFH98a} give a polynomial-time algorithm for matching
games by reducing the successive lexicographic stages to compact linear
systems. The recent survey of \citet{BBJ+23a} leaves three questions
about the nucleon open. Open Problem~5, which the survey traces back to
\citet{FKFH98a}, asks whether the nucleon of a matching game admits a
polynomial-time combinatorial algorithm. Open Problem~7 asks for the
complexity of computing the nucleon of $b$-matching games and of
partitioned matching games, two separate classes.

\paragraph{Our contributions} We resolve all three open problems regarding the nucleon.
\begin{enumerate}
\item \emph{Matching games (Open Problem~5).} For matching games with
 nonnegative rational weights, we give a polynomial-time algorithm that
 uses only maximum-weight matching, shortest paths, minimum-cost flow, and
 exact rational arithmetic (Theorem~\ref{thm:algorithm}). It returns both
 a nucleon allocation and a compact description of all nucleon
 allocations.
\item \emph{$b$-matching games (Open Problem~7, first part).} Computing
 even one nucleon allocation is NP-hard when every edge has weight one and
 every vertex capacity is at most three (Theorem~\ref{thm:hardness}).
\item \emph{Partitioned matching games (Open Problem~7, second part).}
 Computing even one nucleon allocation is NP-hard when every edge has
 weight one and each player owns at most three vertices
 (Theorem~\ref{thm:partitioned}).
\end{enumerate}

{For the algorithm, we retain the structural reduction of
\citet{FKFH98a} and replace its linear-programming feasibility steps with
graph algorithms. We restate their compact-stage description in our
notation and prove only the simplified bound that is used by our new
combinatorial routine. The running time is polynomial in the binary input
length; we do not claim a strongly polynomial bound.}

Both hardness proofs add two players joined by a disjoint unit-weight
edge. Their combined payoff equals one precisely when the original game
has a nonempty core. Thus an exact allocation reveals a hard decision
without any additional coalition optimization. {Unit weights count
two-way exchanges in kidney exchange. Each exchange yields two transplants,
so multiplying all values and allocations by two gives transplant counts
without changing any ratio.
Theorem~\ref{thm:partitioned} therefore implies that exact nucleon-based
allocations for international kidney exchange cannot be computed in
polynomial time unless $\mathrm{P}=\mathrm{NP}$, even if only pairwise
exchanges are allowed and each country owns at most three patient--donor
pairs.} With capacity one, or with
one vertex per player, both classes reduce to matching games, which are
covered by Theorem~\ref{thm:algorithm}. The boundary cases of capacities
at most two and of at most two vertices per player remain open.

For capacities at most two, we sharpen this frontier. Building on the
linear-subspace-avoidance framework of \citet{EE26a}, we show that exact
nucleon computation for $b\leq2$ is polynomial-time Turing reducible to
\SNZC{} (Theorem~\ref{thm:b2-snzc}). \rev{The unrestricted minimum-excess
optimization oracle used by our separation routine is polynomially
equivalent to \SNZC{}. The resulting algorithm is conditional, since the
complexity of \SNZC{} is itself open. We establish no reverse reduction
for the restricted queries arising in actual nucleon stages or for the
two full computation problems.}

Section~\ref{sec:related} places these results in the literature.
Section~\ref{sec:model} defines the model and the nucleon.
Section~\ref{sec:algorithm} gives the combinatorial algorithm, and
Section~\ref{sec:hardness} proves the two hardness results.
{\rev{Section~\ref{sec:capacity-two} relates capacity-two stage separation
to the unrestricted minimum-excess oracle and derives a conditional
algorithm from the framework of \citet{EE26a}.}}

\section{Related work}\label{sec:related}

\paragraph{Matching games}
The assignment game of \citet{ShSh71a} models a two-sided market with
transferable utility. It is the bipartite special case of a matching
game, and its core is nonempty. \citet{DIN99a} study the core of
combinatorial optimization games, including matching games. For an
ordinary matching game, nonnegativity and the edge inequalities
$x_i+x_j\geq w(ij)$ imply every coalition inequality: one sums them
over a maximum-weight matching of the coalition. \citet{BKP12a} give
polynomial-time algorithms for testing core nonemptiness and finding a
core allocation, as well as an $O(n^4)$-time algorithm for the nucleolus
when the core is nonempty.

This simple description does not extend to $b$-matching games, because
several chosen edges can share a vertex. \citet{BKPW18a} show that core
membership and separation are polynomial-time solvable when every
capacity is at most two; \citet{SV25a} give a corrected separation
argument and a compact extended formulation. \citet{BKPW18a} also prove
that core membership is
coNP-complete for unit weights and capacity three, even on bipartite
graphs. \citet[Theorem~1]{BBK+25a} prove that core nonemptiness is
coNP-hard for unit weights and capacities at most three. The same paper
introduces partitioned matching games as a model of international kidney
exchange. It shows that core nonemptiness is polynomial-time decidable
when each player owns at most two vertices, and coNP-hard for unit
weights when players own up to three vertices. Our hardness proofs use
these two nonemptiness results. They must distinguish an empty core from
a nonempty one, so hardness of core membership alone would not suffice.

For the Shapley value of matching games, \citet{AzKe14a} show
\#P-completeness in the unweighted case and give a fully polynomial-time
randomized approximation scheme.

\paragraph{The nucleon and multiplicative stability}
\citet{FaKe93a} study multiplicative relaxations of the core in
combinatorial games. \citet{FKFH98a} introduce the nucleon and prove its
polynomial-time computability in matching games. Their algorithm
discards coalition constraints whose payoff is already fixed and
represents each remaining stage by polynomially many constraints
computed by weighted matching. Testing the feasibility of such a stage
is a linear program. We use the same representation and supply a
combinatorial feasibility routine for the resulting systems of
singleton and pair bounds with an efficiency equality.

Multiplicative core approximation addresses a related, weaker goal.
\citet{Vazi22a} gives a polynomial-time algorithm for a $2/3$-approximate
core allocation in ordinary matching games; the factor is optimal in the
worst case. \citet{XLLF25a} obtain a $2/3$ guarantee in a multiple-partners
model that permits repeated use of an edge. Their convention differs
from ours, which permits each edge at most once. In either setting,
a prescribed approximation factor does not optimize the entire
lexicographic ratio vector.

\paragraph{The nucleolus}
The nucleon is the multiplicative counterpart of the nucleolus of
\citet{Schm69a}. For the additive objective, \citet{SoRa94a} give an
$O(n^4)$-time algorithm for assignment games, and \citet{KePa03a}
give a polynomial-time algorithm for uniform matching games.
\citet{Hard17a} obtains an $O(n^4)$-time combinatorial algorithm for
uniform matching games. \citet{KPT20a} prove that the nucleolus of
a general weighted matching game can be computed in polynomial time
using the ellipsoid method. A combinatorial algorithm for this additive
objective is a separate open question.

For capacitated games, \citet{KTZ24a} prove that deciding whether a given
allocation is the nucleolus is NP-hard for unit weights and capacity
three, even on bipartite graphs of maximum degree seven, and give
polynomial-time algorithms for restricted capacity-two cases. This
additive hardness result does not classify the multiplicative objective.
More recently, \citet{EE26a} reduce later-stage nucleolus separation to
linear-subspace-avoidance optimization and relate its capacity-two case
to \SNZC. Section~\ref{sec:capacity-two} adapts that oracle connection to
the multiplicative stages of the nucleon.
Our three theorems resolve Open Problem~5 and both parts of Open
Problem~7 in the survey of \citet{BBJ+23a}.

\section{Preliminaries}\label{sec:model}

The three-firm example illustrates both the model and its allocation
problem. A coalition's value is the best set of partnerships it can form,
while its payoff is the sum of the amounts assigned to its members.
We now define these quantities.

\subsection{Matching games and allocations}

Let $G=(N,E)$ be a finite simple undirected graph with nonnegative
rational edge weights $w:E\to\R_{\geq0}$, and let $n=|N|$. The vertices
are the \textbf{players}, and a subset of players is a
\textbf{coalition}. For a set of edges $M$, write $w(M)=\sum_{e\in M}w(e)$,
and let $V(M)$ denote the vertices incident to its edges. A
\textbf{matching} is a set of edges with no common endpoints. The
\textbf{matching game} defined by $(G,w)$ has value function
\begin{equation}\label{eq:matching-value}
 v(S)=\max\{w(M): M\text{ is a matching in }G[S]\},
 \qquad S\subseteq N.
\end{equation}
The empty matching is allowed, so $v(\emptyset)=0$ and
$v(\{i\})=0$ for every player $i$.

A capacity function $b:N\to\mathbb{Z}_{\geq1}$ permits several
partnerships at a vertex. A \textbf{$b$-matching} is a subset $M\subseteq E$
in which at most $b_i$ edges meet each vertex $i$. The corresponding
\textbf{$b$-matching game} defines $v(S)$ by maximizing $w(M)$ over
the $b$-matchings in $G[S]$. Each edge may be used at most once.
The choice $b_i=1$ for every player recovers a matching game.

In a \textbf{partitioned matching game}, the players own disjoint sets of
vertices. It is given by a graph $H=(V,E)$ with weights
$w:E\to\R_{\geq0}$ and a partition of $V$ into sets $V_p$, $p\in N$;
player $p$ owns the vertices in $V_p$. The value of a coalition $S$ is
the maximum weight of a matching in the subgraph of $H$ induced by
$\bigcup_{p\in S}V_p$. The \textbf{width} of the game is
$\max_{p\in N}|V_p|$. Width one recovers a matching game. Unlike in the
other two classes, a single player can have positive value.
A game of any of these three kinds is \textbf{uniform} when every edge has
weight one.

Write $W=v(N)$ for the value of the \textbf{grand coalition}.
For a vector $x\in\R^N$, let $x(S)=\sum_{i\in S}x_i$.
An \textbf{imputation} is a nonnegative efficient allocation:
\begin{equation}\label{eq:imputation}
 x_i\geq0\quad(i\in N),\qquad x(N)=W.
\end{equation}
We write $I$ for the set of imputations. In matching and $b$-matching
games, singleton values are zero, so nonnegativity is individual
rationality. Following \citet{FKFH98a}, we use the same set $I$ for every
game with nonnegative values; a singleton with positive value is then
compared through its ratio like any other coalition. The
\textbf{core} is
\begin{equation}\label{eq:core}
 \core(N,v)=\{x\in I:
                   x(S)\geq v(S)\text{ for every }S\subseteq N\}.
\end{equation}
Thus a core allocation pays every coalition at least its outside option.

\subsection{The nucleon}

The nucleon ranks coalitions by the fraction of their value they receive.
Coalitions with value zero need no ratio comparison, since
nonnegative payoffs already meet their core constraints.

\begin{definition}[Nucleon]\label{def:nucleon}
For an imputation $x$, let $\theta(x)$ be the list of the ratios
\[
 \frac{x(S)}{v(S)}
 \qquad(\emptyset\ne S\subsetneq N,\ v(S)>0),
\]
with one entry for each such coalition, sorted in nondecreasing order.
The \textbf{nucleon} $\nucleon(N,v)$ is the set of imputations $x$ for
which $\theta(x)$ is lexicographically maximal.
\end{definition}

Lexicographic maximization first improves the smallest entry. Subject to
that entry being optimal, it improves the second, and so on.
Equivalently, zero-valued proper coalitions can be assigned ratio
$+\infty$~\citep{FKFH98a}. The grand coalition is excluded because its
ratio is always one and conveys no comparison between allocations.
If no proper coalition has positive value, the nucleon is the entire
imputation set.

When positive-valued proper coalitions exist, the first optimization is
\begin{equation}\label{eq:first-stage}
 \max\ \alpha\quad\text{subject to }\eqref{eq:imputation},\qquad
 x(S)\geq\alpha v(S)
 \quad(\emptyset\ne S\subsetneq N,\ v(S)>0).
\end{equation}
This is the \textbf{multiplicative least-core} problem; we write
$\alpha_1(N,v)$ for its optimal value, and set $\alpha_1(N,v)=+\infty$
when no proper coalition has positive value. The nucleon refines its optimal
set by further ratio comparisons. Unlike the nucleolus, the nucleon can
contain several allocations.

{The first stage gives the nucleon two guarantees. If the core is
nonempty, some imputation gives every coalition its full value, so
$\alpha_1(N,v)\geq1$ and every nucleon allocation lies in the core.
\rev{When positive-valued proper coalitions exist, every nucleon allocation
gives each such coalition at least the fraction $\alpha_1(N,v)$ of its
value, and no imputation guarantees them a larger fraction. If $W>0$, the
best common guarantee for all coalitions, including $N$, is
$\min\{1,\alpha_1(N,v)\}$, with $\min\{1,+\infty\}=1$ when no proper
coalition has positive value.} For matching games $\alpha_1(N,v)\geq2/3$
\citep{FaKe93a}, so every nucleon allocation lies in the $2/3$-approximate
core.}

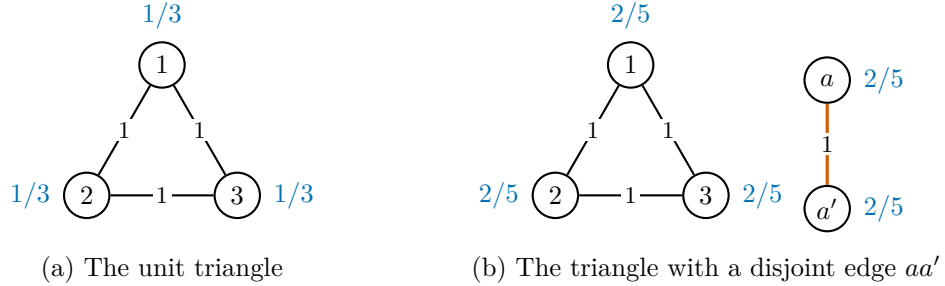
\begin{figure}[t]
\centering
\begin{tikzpicture}
 % (a) unit triangle
 \begin{scope}
  \node[player] (t1) at (90:1.15) {$1$};
  \node[player] (t2) at (210:1.15) {$2$};
  \node[player] (t3) at (330:1.15) {$3$};
  \draw[thick] (t1) -- node[wt] {$1$} (t2)
                    -- node[wt] {$1$} (t3)
                    -- node[wt] {$1$} (t1);
  \node[payoff,above=1pt] at (t1.north) {$1/3$};
  \node[payoff,left=1pt]  at (t2.west)  {$1/3$};
  \node[payoff,right=1pt] at (t3.east)  {$1/3$};
  \node[font=\small] at (0,-1.55) {(a) The unit triangle};
 \end{scope}
 % (b) triangle plus a disjoint edge
 \begin{scope}[xshift=6.2cm]
  \node[player] (u1) at (90:1.15) {$1$};
  \node[player] (u2) at (210:1.15) {$2$};
  \node[player] (u3) at (330:1.15) {$3$};
  \draw[thick] (u1) -- node[wt] {$1$} (u2)
                    -- node[wt] {$1$} (u3)
                    -- node[wt] {$1$} (u1);
  \node[payoff,above=1pt] at (u1.north) {$2/5$};
  \node[payoff,left=1pt]  at (u2.west)  {$2/5$};
  \node[payoff,right=1pt] at (u3.east)  {$2/5$};
  \node[player] (a) at (2.6,0.95) {$a$};
  \node[player] (b) at (2.6,-0.75) {$a'$};
  \draw[very thick,corange] (a) -- node[wt,text=black] {$1$} (b);
  \node[payoff,right=1pt] at (a.east) {$2/5$};
  \node[payoff,right=1pt] at (b.east) {$2/5$};
  \node[font=\small,align=center] at (1.0,-1.55)
       {(b) The triangle with a disjoint edge $aa'$};
 \end{scope}
\end{tikzpicture}
\caption{Two small matching games; edge labels are weights and blue
numbers are the unique nucleon allocations. (a)~Every pair receives
$2/3$ of its value, and the core is empty (Example~\ref{ex:triangle}).
(b)~After a disjoint unit-weight edge is added, the first stage gives
the triangle's pairs and the edge $aa'$ the same ratio $4/5$. The edge
total $x_a+x_{a'}=4/5<1$ records that the triangle has an empty core
(Section~\ref{sec:hardness}).}
\label{fig:triangle}
\end{figure}

\begin{example}[A triangle and a single partnership]\label{ex:triangle}
In the unit-weight triangle of Figure~\ref{fig:triangle}(a), $W=1$ and
each pair has value one. The three pair payoffs sum to $2x(N)=2$, so
their smallest ratio is at most $2/3$. The allocation $(1/3,1/3,1/3)$
attains this bound. At an optimum all three pair sums must equal $2/3$,
which fixes this allocation uniquely. It is therefore the nucleon,
although the core is empty.

For two players joined by a unit-weight edge, both proper coalitions
have value zero. Every nonnegative split of one belongs to the nucleon.
This explains why our algorithm returns a description of the full set
as well as an allocation.
\end{example}

The nucleon is never empty; Lemma~\ref{lem:stages} below shows that it
is a polytope. We study the task of returning an exact rational member of
the nucleon from $(G,w)$, from $(G,b,w)$, or from a partitioned instance
$(H,w,(V_p)_{p\in N})$. All input numbers are encoded in binary. A
compact description means a system with polynomially many constraints
and polynomial encoding length. Our combinatorial algorithm uses graph
routines and exact arithmetic, without the ellipsoid method or a
general-purpose linear programming solver.

\section{A combinatorial algorithm for matching games}\label{sec:algorithm}

The main obstacle is not finding the value of a coalition. It is handling
the exponentially many coalitions while retaining the efficiency
equality. We follow the stages of \citet{FKFH98a}: each stage maximizes
the smallest ratio among the coalitions whose payoff is not yet fixed.
For matching games, each stage can be written with one bound per pair or
singleton of players. We show that such bounds, together with efficiency,
can be handled by shortest paths and minimum-cost flow.

\begin{theorem}\label{thm:algorithm}
For matching games with nonnegative rational edge weights, a nucleon
allocation and a compact description of the entire nucleon can be
computed in polynomial time using maximum-weight matching, shortest
paths, minimum-cost flow, and exact rational arithmetic.
\end{theorem}

The theorem provides a graph-based implementation of the previously
known polynomial-time computation. Section~\ref{sec:stages} describes
the stages for an arbitrary nonnegative game, and
Section~\ref{sec:compact} shows that for matching games every stage
involves only pair and singleton sums. Section~\ref{sec:pair} solves
such systems with a prescribed total, Section~\ref{sec:exact} makes every
optimization exact, and Section~\ref{sec:proof} assembles
Algorithm~\ref{alg:nucleon}.

\subsection{Successive stages}\label{sec:stages}

Throughout this subsection $(N,v)$ is any game with $v\geq0$ and
$W=v(N)>0$. A coalition $S$ is \textbf{fixed} on a set $P\subseteq I$
if $x(S)$ takes the same value at every $x\in P$. The stages form a
decreasing sequence $P_0\supseteq P_1\supseteq\cdots$ of polytopes.
Start with $P_0=I$. For $k=1,2,\ldots$, let $U_{k-1}$ be the set of
proper coalitions $S$ with $v(S)>0$ that are not fixed on $P_{k-1}$.
If $U_{k-1}=\emptyset$, stop and set $K=k-1$. Otherwise let
\begin{align}
 \alpha_k&=\max_{x\in P_{k-1}}\ \min_{S\in U_{k-1}}\frac{x(S)}{v(S)},
 \label{eq:stage-value}\\
 P_k&=\{x\in P_{k-1}: x(S)\geq\alpha_k v(S)\ \text{for all }S\in U_{k-1}\}.
 \label{eq:stage-set}
\end{align}
The maximum exists because $P_{k-1}$ is a nonempty polytope. This
\emph{accelerated} scheme skips every comparison whose outcome is
already decided; it is the multiplicative analogue of the sequential
programs for the nucleolus \citep{MPS79a}. For $1\leq k\leq K+1$ and
$\alpha\geq0$, let
\begin{equation}\label{eq:stage-region}
 R_k(\alpha)=\{x\in P_{k-1}: x(S)\geq\alpha v(S)\ \text{for all }
              S\in U_{k-1}\},
\end{equation}
Thus $R_{K+1}(\alpha)=P_K$ for every $\alpha$, while for $k\leq K$ we
have $R_k(\alpha)\neq\emptyset$ exactly when $\alpha\leq\alpha_k$, and
$P_k=R_k(\alpha_k)$.

\begin{lemma}\label{lem:stages}
Set $\alpha_0=0$. The scheme has the following properties.
\begin{enumerate}
\item[(a)] $R_k(\alpha_{k-1})=P_{k-1}$ for $1\leq k\leq K+1$; in
 particular $\alpha_k\geq\alpha_{k-1}$ for $1\leq k\leq K$.
\item[(b)] For $1\leq k\leq K$, some coalition in $U_{k-1}$ is fixed on
 $P_k$. Hence $\dim P_k<\dim P_{k-1}$, and the scheme stops after
 $K\leq n-1$ stages.
\item[(c)] $\nucleon(N,v)=P_K$.
\end{enumerate}
\end{lemma}

\begin{proof}
(a) For $k=1$ the claim holds because $x\geq0$. For $k\geq2$, a
coalition fixed on $P_{k-2}$ remains fixed on its subset $P_{k-1}$, so
$U_{k-1}\subseteq U_{k-2}$. By \eqref{eq:stage-set}, every
$x\in P_{k-1}$ satisfies $x(S)\geq\alpha_{k-1}v(S)$ for all
$S\in U_{k-2}$, and in particular for all $S\in U_{k-1}$.

(b) Suppose instead that every $S\in U_{k-1}$ has a point
$x^S\in P_k$ with $x^S(S)>\alpha_k v(S)$. The average $\bar x$ of these
points lies in the convex set $P_k\subseteq P_{k-1}$, and
$\bar x(S)>\alpha_k v(S)$ for every $S\in U_{k-1}$. This contradicts
the maximality of $\alpha_k$ in \eqref{eq:stage-value}. Thus some
$S\in U_{k-1}$ satisfies $x(S)=\alpha_k v(S)$ throughout $P_k$. The
affine function $x\mapsto x(S)$ is then constant on the affine hull of
$P_k$ but not on that of $P_{k-1}$, so the dimension drops. Since
$\dim P_0\leq n-1$ and every coalition is fixed on a single point, there
are at most $n-1$ stages.

(c) We use one property of sorted lists: adding the same entries to two
lists does not change which of them is lexicographically larger. Indeed,
compare the lists through their counting functions
$t\mapsto|\{\text{entries}\leq t\}|$; the first list is larger exactly
when, at the smallest $t$ where the two functions differ, it has fewer
entries. Adding common entries adds the same amount to both counting
functions.

Coalitions that are fixed on $P_{k-1}$ contribute the same ratios to
$\theta(x)$ for all $x\in P_{k-1}$. Hence, on $P_{k-1}$, comparing
$\theta(x)$ and $\theta(y)$ amounts to comparing the sorted ratio lists
of the coalitions in $U_{k-1}$. Their first entries are the minima in
\eqref{eq:stage-value}. So every $x\in P_{k-1}\setminus P_k$ is beaten
by every $y\in P_k$, and the lexicographic maximizers of $\theta$ over
$P_{k-1}$ are exactly those over $P_k$. On $P_K$ every coalition with
positive value is fixed, so $\theta$ is constant there. Consequently
$\nucleon(N,v)$, the set of maximizers over $P_0=I$, equals $P_K$.
\end{proof}

\subsection{Compact stages for matching games}\label{sec:compact}

From now on $(N,v)$ is a matching game with $n\geq3$ and $W>0$. Complete
the graph by adding zero-weight edges, so that every pair
$A=\{i,j\}$ has a weight $w(A)=w(ij)\geq0$; this does not change $v$.
For a singleton $A$, put $w(A)=0$. Consider $1\leq k\leq K+1$. Let
$E_{k-1}$ be the set of pairs fixed on $P_{k-1}$, and let $F_{k-1}$ be
the set of players whose payoff is fixed on $P_{k-1}$. Write $c(A)$ for
the common value of $x(A)$ on $P_{k-1}$ when the pair or singleton $A$ is
fixed, and $c(i)=c(\{i\})$. The remaining pairs and singletons are
\textbf{unfixed}; we collect them in $\mathcal{A}_{k-1}$. The equalities
known so far define
\begin{equation}\label{eq:domain}
 D_{k-1}=\Bigl\{x\geq0:\ \begin{aligned}
 &x(N)=W,\quad x(e)=c(e)\ (e\in E_{k-1}),\\
 &x_i=c(i)\ (i\in F_{k-1})\end{aligned}\Bigr\}.
\end{equation}
For $A\in\mathcal{A}_{k-1}$ and $\alpha\geq0$, let
\begin{equation}\label{eq:mu}
 \mu_A(\alpha)=\max\Bigl\{\sum_{e\in M}\bigl(\alpha w(e)-c(e)\bigr):
 M\subseteq E_{k-1}\text{ a matching avoiding }A\Bigr\},
\end{equation}
where $M$ \emph{avoids} $A$ when $V(M)\cap A=\emptyset$.
The empty matching shows $\mu_A(\alpha)\geq0$, and one maximum-weight
matching computation in the graph of fixed pairs on $N\setminus A$
evaluates $\mu_A(\alpha)$. Note that $P_{k-1}\subseteq D_{k-1}$, because
every equality defining $D_{k-1}$ holds on $P_{k-1}$.

{The structural reduction in this subsection is due to
\citet[Propositions~4.1--4.2]{FKFH98a}. We record it in our notation because
it is the interface used by the combinatorial algorithm. Our only
modification is the explicit bound
\begin{equation}\label{eq:compact-bound}
 x(A)\geq\alpha w(A)+\mu_A(\alpha),
 \qquad A\in\mathcal A_{k-1},
\end{equation}
which avoids constructing and re-optimizing over an auxiliary coalition.

\begin{proposition}[Compact stages; \cite{FKFH98a}]\label{prop:compact}
Let $(N,v)$ be a matching game with $n\geq3$ and $W>0$, and let
$1\leq k\leq K+1$.
\begin{enumerate}
\item[(a)] Every coalition fixed on $P_{k-1}$ has the same payoff at all
points of $D_{k-1}$.
\item[(b)] For every $\alpha\geq\alpha_{k-1}$,
\[
R_k(\alpha)=\{x\in D_{k-1}:x(A)\geq \alpha w(A)+\mu_A(\alpha)
\text{ for every }A\in\mathcal A_{k-1}\}.
\]
\end{enumerate}
Moreover, $\nucleon(N,v)=D_K$.
\end{proposition}

\begin{proof}
Part~(a), the identity $\nucleon(N,v)=D_K$, and a polynomial-size
representation of each $R_k(\alpha)$ are exactly the structural conclusions
of \citet[Propositions~4.1--4.2]{FKFH98a}, translated to the accelerated
stages of Section~\ref{sec:stages}. Their constraint associated with an
unfixed pair or singleton $A$ is
\[
 x(S_A)\geq \alpha v(S_A),\qquad
 S_A=A\cup V(M_A),
\]
where $M_A\subseteq E_{k-1}$ avoids $A$ and maximizes
$\sum_{e\in M_A}(\alpha w(e)-c(e))=\mu_A(\alpha)$.
On $D_{k-1}$, $x(S_A)=x(A)+c(M_A)$, while
$v(S_A)\geq w(A)+w(M_A)$. Hence their constraint implies
\eqref{eq:compact-bound}.

It remains only to show that the simpler bounds are jointly sufficient.
Let $x\in D_{k-1}$ satisfy \eqref{eq:compact-bound} for every unfixed
$A$, and let $S\in U_{k-1}$. Choose a maximum-weight matching $M$ in
$G[S]$. The edges of $M$ together with the unmatched vertices partition
$S$. Since $S$ is not fixed on $P_{k-1}$, at least one such part, say
$A$, is unfixed. The fixed edges of $M$ form a matching in $E_{k-1}$
avoiding $A$, and therefore
\[
 \sum_{e\in M\cap E_{k-1}}\bigl(c(e)-\alpha w(e)\bigr)
 \geq-\mu_A(\alpha).
\]
The part $A$ contributes at least $\mu_A(\alpha)$ by
\eqref{eq:compact-bound}. Every other unfixed matched edge $e$ contributes
$x(e)-\alpha w(e)\geq\mu_e(\alpha)\geq0$, and every unmatched vertex
contributes a nonnegative amount. Summing over the partition of $S$ gives
\[
 x(S)-\alpha v(S)\geq0.
\]
\rev{These inequalities also preserve all earlier stages. Indeed, every
coalition constrained at an earlier stage is either fixed on $P_{k-1}$,
in which case part~(a) preserves its payoff on $D_{k-1}$, or belongs to
$U_{k-1}$, in which case the new bound is at least as strong because
$\alpha\geq\alpha_{k-1}$. Thus $x\in P_{k-1}$.}
{Conversely, each constraint of \citet{FKFH98a} is itself one of these
coalition constraints, or holds trivially. Indeed, $x(S_A)=x(A)+c(M_A)$ on
$D_{k-1}\supseteq P_{k-1}$ and $x(A)$ is not constant on $P_{k-1}$, so
$S_A$ is not fixed on $P_{k-1}$; in particular $S_A\neq N$. Hence either
$S_A\in U_{k-1}$, or $v(S_A)=0$ and the constraint reads $x(S_A)\geq0$.
Therefore, on $D_{k-1}$, the bounds \eqref{eq:compact-bound} and the
system of \citet{FKFH98a} imply each other, so they describe the same set
$R_k(\alpha)$.} This proves part~(b). The final statement follows from
part~(a) and the terminal-stage characterization in
Lemma~\ref{lem:stages}(c).
\end{proof}}

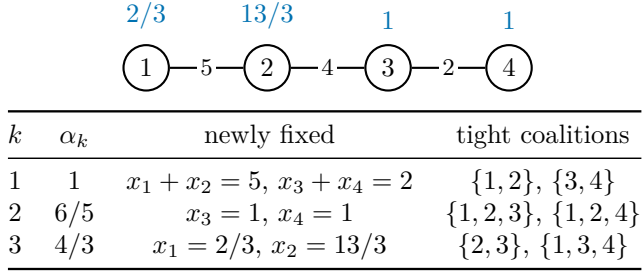
\begin{figure}[t]
\centering
\begin{tikzpicture}[x=1.6cm]
 \foreach \i in {1,...,4} \node[player] (p\i) at (\i,0) {$\i$};
 \draw[thick] (p1) -- node[wt] {$5$} (p2)
                   -- node[wt] {$4$} (p3)
                   -- node[wt] {$2$} (p4);
 \node[payoff,above=2pt] at (p1.north) {$2/3$};
 \node[payoff,above=2pt] at (p2.north) {$13/3$};
 \node[payoff,above=2pt] at (p3.north) {$1$};
 \node[payoff,above=2pt] at (p4.north) {$1$};
\end{tikzpicture}

\medskip
{\small
\begin{tabular}{@{}cccc@{}}
\toprule
$k$ & $\alpha_k$ & newly fixed & tight coalitions\\
\midrule
1 & $1$ & $x_1+x_2=5$, $x_3+x_4=2$ & $\{1,2\}$, $\{3,4\}$\\
2 & $6/5$ & $x_3=1$, $x_4=1$ & $\{1,2,3\}$, $\{1,2,4\}$\\
3 & $4/3$ & $x_1=2/3$, $x_2=13/3$ & $\{2,3\}$, $\{1,3,4\}$\\
\bottomrule
\end{tabular}}
\caption{The weighted path of Example~\ref{ex:path}, with its nucleon
allocation in blue, and the three stages of the algorithm.}
\label{fig:path}
\end{figure}

\begin{example}[A weighted path]\label{ex:path}
Consider the path $1$--$2$--$3$--$4$ with weights $5$, $4$ and $2$
(Figure~\ref{fig:path}), so $W=7$. At stage~1 nothing is fixed and
$\mu_A\equiv0$, so the bounds are $x_i+x_j\geq\alpha w(ij)$ and
$x_i\geq0$. Adding the bounds for $\{1,2\}$ and $\{3,4\}$ gives
$7=x(N)\geq7\alpha$, and $x=(1,4,0,2)$ attains $\alpha=1$. Hence
$\alpha_1=1$, and both bounds are tight on $P_1$, which fixes
$x_1+x_2=5$ and $x_3+x_4=2$. At stage~2 the only fixed pair avoiding
player~3 is $\{1,2\}$, so the bound for $A=\{3\}$ is
$x_3\geq\mu_{\{3\}}(\alpha)=5\alpha-5$; this is the constraint of the
coalition $\{1,2,3\}$. Symmetrically $x_4\geq5\alpha-5$. With
$x_3+x_4=2$ these give $\alpha_2=6/5$, the remaining bounds can still
be met, and $x_3=x_4=1$ becomes fixed. At stage~3,
the bounds $x_1\geq2\alpha-2$ (from $A=\{1\}$ and the fixed pair
$\{3,4\}$) and $x_2+x_3\geq4\alpha$, together with $x_1+x_2=5$ and
$x_3=1$, give $\alpha_3=4/3$ and fix the allocation
$(\tfrac23,\tfrac{13}3,1,1)$, which is the nucleon. Later stages may
have levels above one, because coalitions such as $\{2,3\}$ can then
receive more than their value.
\end{example}

\subsection{Pair bounds with a prescribed total}\label{sec:pair}

By Proposition~\ref{prop:compact}, every stage system consists of bounds
on single coordinates and on pair sums, together with efficiency. Without
efficiency, such bounds are difference constraints on signed copies of
the variables. To impose the total, we compute the smallest and largest
feasible totals and interpolate between them.

\begin{lemma}\label{lem:pair}
Let $C\geq0$ be rational. Consider a rational system $\Sigma$ of lower and
upper bounds on single coordinates $x_i$ and on pair sums $x_i+x_j$,
including equalities of these types and the bounds $0\leq x_i\leq C$
for every $i$. One Bellman--Ford computation and two minimum-cost flow
computations decide whether $\Sigma$ has a solution with $\sum_ix_i=C$,
and find one if it exists.
\end{lemma}

\begin{proof}
\step{Represent the bounds by difference constraints.}
For every player $i$ introduce two nodes $i^+$ and $i^-$ with potentials
$p(i^+)$ and $p(i^-)$, intended to represent $x_i$ and $-x_i$. Each bound
becomes two difference constraints:
\begin{align}
 x_i+x_j\leq u
 &:\quad p(i^+)-p(j^-)\leq u,\quad
          p(j^+)-p(i^-)\leq u,\label{eq:upper-pair}\\
 x_i+x_j\geq\ell
 &:\quad p(j^-)-p(i^+)\leq-\ell,\quad
          p(i^-)-p(j^+)\leq-\ell,\label{eq:lower-pair}\\
 \ell_i\leq x_i\leq u_i
 &:\quad p(i^+)-p(i^-)\leq2u_i,\quad
          p(i^-)-p(i^+)\leq-2\ell_i.\label{eq:individual-potentials}
\end{align}
An equality contributes both its upper and lower bound. A constraint
$p(h)-p(g)\leq\gamma$ is an arc from $g$ to $h$ of length $\gamma$
(Figure~\ref{fig:network}). Such a system of difference constraints is
feasible exactly when the network has no negative cycle, and then
Bellman--Ford supplies feasible potentials. If Bellman--Ford finds a
negative cycle, then $\Sigma$ has no solution by the correspondence
below, and we stop.

The two systems are equivalent, and the correspondence preserves totals.
Potentials of the form $p(i^+)=x_i$, $p(i^-)=-x_i$ are feasible exactly
when $x$ solves $\Sigma$, because each translated constraint then reads
as the original bound. Moreover, if $p$ is feasible, so is its reflection
$\bar p(i^+)=-p(i^-)$, $\bar p(i^-)=-p(i^+)$: the reflection maps each
pair of constraints in \eqref{eq:upper-pair}--\eqref{eq:individual-potentials}
to itself. By convexity $(p+\bar p)/2$ is feasible, and it has the form
$(x,-x)$ with
\begin{equation}\label{eq:recover-coordinate}
 x_i=\frac{p(i^+)-p(i^-)}{2}.
\end{equation}
Hence this $x$ solves $\Sigma$, and its total is
$\sum_ix_i=\tfrac12\sum_i\bigl(p(i^+)-p(i^-)\bigr)$.

\step{Find the range of feasible totals.}
Let $L$ and $U$ be the smallest and largest total of a solution of
$\Sigma$. By the correspondence, $2U$ is the maximum of
$\sum_i\bigl(p(i^+)-p(i^-)\bigr)$ over feasible potentials $p$, and $-2L$
is the maximum of $\sum_i\bigl(p(i^-)-p(i^+)\bigr)$. Both objectives have
the form $d^{\mathsf T}p$ for coefficients $d$ that sum to zero, and for
such $d$ linear programming duality gives
\begin{equation}\label{eq:transshipment}
 \begin{split}
 \max\{d^{\mathsf T}p: p\text{ feasible}\}
 &=\min\Bigl\{\sum_{gh}\gamma_{gh}f_{gh}:\ f\geq0,\\
 &\qquad\quad\operatorname{in}_h(f)-\operatorname{out}_h(f)=d_h\
   \text{for all nodes }h\Bigr\},
 \end{split}
\end{equation}
where $\gamma_{gh}$ is the length of arc $gh$, and
$\operatorname{in}_h(f)$ and $\operatorname{out}_h(f)$ are the flow into
and out of node $h$. The right side is an uncapacitated minimum-cost
flow problem. Taking $d(i^+)=1$ and $d(i^-)=-1$ gives optimum $2U$:
one unit of flow leaves each node $i^-$ and one unit enters each node
$i^+$. Reversing these signs gives optimum $-2L$. The arcs between $i^-$
and $i^+$ coming from $0\leq x_i\leq C$ give feasible flows, and feasible
potentials exist, so both linear programs have optima. A minimum-cost flow algorithm
returns an optimal flow together with optimal node potentials, and these
potentials attain the maximum \citep{Orli93a}. For example, the
successive shortest path method, started from the Bellman--Ford
potentials, needs $n$ Dijkstra computations, one for each unit of flow. Formula
\eqref{eq:recover-coordinate} turns the optimal potentials into solutions
$y$ and $z$ of $\Sigma$ with totals $L$ and $U$.

\step{Impose the prescribed total.}
A solution with total $C$ exists exactly when $L\leq C\leq U$.
If $L=U$, then $y$ itself has total $C$. Otherwise the convex combination
\[
 \frac{U-C}{U-L}\,y+\frac{C-L}{U-L}\,z
\]
has total $C$ and satisfies every bound and equality of $\Sigma$.
\end{proof}

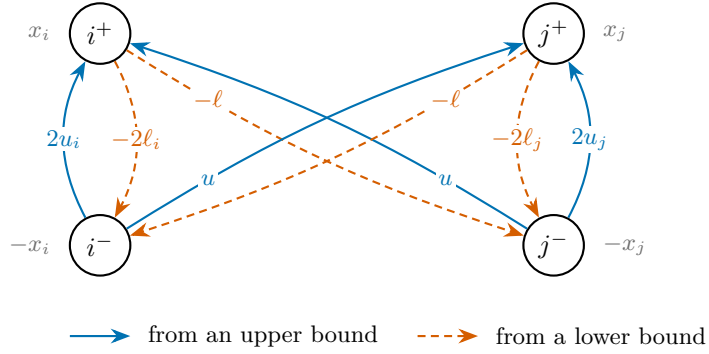
\begin{figure}[t]
\centering
\begin{tikzpicture}[>={Stealth[length=2.6mm]},
   pot/.style={circle,draw,thick,fill=white,inner sep=1pt,
               minimum size=8mm,font=\small},
   ub/.style={->,thick,cblue},
   lb/.style={->,thick,corange,densely dashed},
   lab/.style={font=\footnotesize,fill=white,inner sep=1pt}]
 \node[pot] (ip) at (0,2.8) {$i^+$};
 \node[pot] (im) at (0,0)   {$i^-$};
 \node[pot] (jp) at (6,2.8) {$j^+$};
 \node[pot] (jm) at (6,0)   {$j^-$};
 \node[font=\footnotesize,cgray,left=3pt]  at (ip.west) {$x_i$};
 \node[font=\footnotesize,cgray,left=3pt]  at (im.west) {$-x_i$};
 \node[font=\footnotesize,cgray,right=3pt] at (jp.east) {$x_j$};
 \node[font=\footnotesize,cgray,right=3pt] at (jm.east) {$-x_j$};
 % individual bounds
 \draw[ub] (im) to[bend left=28]  node[lab,pos=0.5] {$2u_i$} (ip);
 \draw[lb] (ip) to[bend left=28]  node[lab,pos=0.5] {$-2\ell_i$} (im);
 \draw[ub] (jm) to[bend right=28] node[lab,pos=0.5] {$2u_j$} (jp);
 \draw[lb] (jp) to[bend right=28] node[lab,pos=0.5] {$-2\ell_j$} (jm);
 % pair bounds
 \draw[ub] (jm) to[bend right=7] node[lab,pos=0.2] {$u$} (ip);
 \draw[lb] (ip) to[bend right=7] node[lab,pos=0.2] {$-\ell$} (jm);
 \draw[ub] (im) to[bend left=7]  node[lab,pos=0.2] {$u$} (jp);
 \draw[lb] (jp) to[bend left=7]  node[lab,pos=0.2] {$-\ell$} (im);
 % legend
 \begin{scope}[shift={(-0.4,-1.2)}]
  \draw[ub] (0,0) -- (0.8,0);
  \node[font=\footnotesize,anchor=west] at (0.9,0) {from an upper bound};
  \draw[lb] (4.6,0) -- (5.4,0);
  \node[font=\footnotesize,anchor=west] at (5.5,0) {from a lower bound};
 \end{scope}
\end{tikzpicture}
\caption{The network in the proof of Lemma~\ref{lem:pair} for two players
with bounds $\ell_i\leq x_i\leq u_i$, $\ell_j\leq x_j\leq u_j$ and
$\ell\leq x_i+x_j\leq u$. An arc from $g$ to $h$ of length $\gamma$
encodes $p(h)-p(g)\leq\gamma$. Exchanging the two rows and reversing every arc maps
the network to itself; this symmetry turns any feasible potentials into
potentials of the form $(x,-x)$. To find the largest total, one unit of
flow leaves every node $i^-$ and one unit enters every node $i^+$.}
\label{fig:network}
\end{figure}

\subsection{Exact arithmetic}\label{sec:exact}

The compact system \eqref{eq:compact-bound} changes with $\alpha$, because
its matchings do. Nevertheless, $R_k(\alpha)$ shrinks as $\alpha$ grows,
so the stage optimum can be found by binary search. The following size
bound, a version of Proposition~3.1 of \citet{FKFH98a}, makes the search
exact. Scaling all weights by a positive integer scales allocations by
the same factor and leaves every ratio unchanged, so we may assume that
the weights are integers. Let $\Omega=1+\sum_{e\in E}w(e)$, so every
coalition value is less than $\Omega$, and let
\[
 \Delta=(n+\Omega)^{2n}.
\]
Its encoding length $O(n\log(n+\Omega))$ is polynomial in the input length.

\begin{lemma}\label{lem:precision}
Let the weights be integers. For $1\leq k\leq K$, the level $\alpha_k$ is
a fraction with denominator at most $\Delta$, and every vertex of $P_k$ is
a rational vector whose coordinates have a common denominator at most
$\Delta$. Consequently, the stored values $c(A)$ and the minimum and
maximum of every pair or singleton sum over $P_k$ are fractions with
denominator at most $\Delta$. All these numbers lie in $[0,W]$.
\end{lemma}

\begin{proof}
Consider the polyhedron $Q_k$ of all pairs $(x,a)\in\R^N\times\R^k$
with $x\in I$ and $x(S)\geq a_jv(S)$ for all $j\leq k$ and
$S\in U_{j-1}$. It contains no line, because $x$ ranges over the bounded
set $I$ and each $a_j$ is bounded above.

We claim that the points of $Q_k$ with $a=(\alpha_1,\ldots,\alpha_k)$
are exactly $P_k\times\{(\alpha_1,\ldots,\alpha_k)\}$, and that these
are the points of $Q_k$ that lexicographically maximize
$(a_1,\ldots,a_k)$. Indeed, if $(x,a)\in Q_k$ has $a_j=\alpha_j$ for
all $j<i$, then $x\in P_{i-1}$ and therefore $a_i\leq\alpha_i$ by
\eqref{eq:stage-value}; equality is attained by every point of $P_k$.
Lexicographic maximization of linear functions over a polyhedron yields
a face, so $F=P_k\times\{(\alpha_1,\ldots,\alpha_k)\}$ is a face of $Q_k$.
Every vertex of $F$ is a vertex of $Q_k$, and $F$ has a vertex because
$Q_k$ contains no line. The vertices of $P_k$ are the $x$-parts of the
vertices of $F$.

A vertex of $Q_k$ is the unique solution of $m=n+k$ linearly
independent rows of the defining system, holding with equality; here
$m\leq2n-1$ because $k\leq K\leq n-1$ by Lemma~\ref{lem:stages}(b). Each row
has at most $n$ coefficients equal to $1$ and at most one coefficient
$-v(S)$ with $v(S)<\Omega$, so its Euclidean norm is at most $n+\Omega$. By
Cramer's rule, all coordinates of the vertex have the common denominator
$|\det B|$ of the chosen $m\times m$ submatrix $B$, and Hadamard's
inequality gives $|\det B|\leq(n+\Omega)^{m}\leq\Delta$. This bounds the
denominators of $\alpha_k$ and of the vertices of $P_k$. A pair or
singleton sum attains its extremes over $P_k$ at vertices, and a fixed
value $c(A)$ equals its value at any vertex. A sum of two coordinates
with a common denominator has the same denominator. Finally, every
$S\in U_{k-1}$ has an integer value $v(S)\geq1$ and $x(S)\leq W$, so
$0\leq\alpha_k\leq W$.
\end{proof}

Two different fractions with denominators at most $\Delta$ differ by at
least $1/\Delta^2$. Hence binary search with a feasibility test locates
$\alpha_k$ exactly. Keep an interval whose lower end passes the test and
whose upper end fails it, and halve it until its length is below
$1/\Delta^2$. Then $\alpha_k$ is the unique fraction with denominator at
most $\Delta$ in the interval, namely the fraction with the smallest
denominator there, which the continued-fraction expansions of the
endpoints produce \citep{Sch86a}. The same search, applied to a test with
an added bound $x(A)\leq\tau$ or $x(A)\geq\tau$, finds the minimum and
maximum of $x(A)$ over $P_k$. Each search uses polynomially many tests on
numbers of polynomial encoding length.

\subsection{The algorithm}\label{sec:proof}

For a stage $k$ and $\alpha\geq0$, let $\textsc{Test}_k(\alpha)$ compute
$\mu_A(\alpha)$ for every $A\in\mathcal{A}_{k-1}$ by maximum-weight
matching, and then decide with Lemma~\ref{lem:pair} (with $C=W$) whether
some $x\in D_{k-1}$ satisfies $x(A)\geq\alpha w(A)+\mu_A(\alpha)$ for all
$A\in\mathcal{A}_{k-1}$; the redundant bounds $x_i\leq W$ are added to
meet the hypothesis of Lemma~\ref{lem:pair}. By
Proposition~\ref{prop:compact}(b), for $\alpha\geq\alpha_{k-1}$ the test
succeeds exactly when $R_k(\alpha)\neq\emptyset$.

Let $w_{\min}$ be the smallest positive edge weight, and let
\begin{equation}\label{eq:alpha-bar}
 \bar\alpha=\frac{W}{w_{\min}}+1 .
\end{equation}
Every coalition $S$ with $v(S)>0$ satisfies $v(S)\geq w_{\min}$, so every
imputation $x$ has $x(S)\leq W<\bar\alpha v(S)$. Thus no coalition of
positive value can reach ratio $\bar\alpha$. The algorithm uses
$\bar\alpha$ as the upper end of each binary search, and the test at
$\bar\alpha$ detects when no stage remains.
Algorithm~\ref{alg:nucleon} summarizes the method.
\rev{It records the integer scaling factor $\lambda$ and returns the
allocation and polytope in the original units. In the returned description,
the efficiency constant $W$ and every fixed value $c(A)$ are divided by
$\lambda$.}

\begin{algorithm}[t]
\caption{Nucleon of a matching game $(G,w)$}\label{alg:nucleon}
\begin{algorithmic}[1]
\State \rev{$\lambda\gets$ least common multiple of the denominators of $w$}
\State \rev{$w\gets\lambda w$; $W\gets$ maximum weight of a matching in $G$}
\State \rev{\textbf{if} $W=0$ \textbf{then return} $\{0\}$ and the zero allocation}
\State \rev{\textbf{if} $n\leq2$ \textbf{then return} $\lambda^{-1}I$ and any member of it}
\State \rev{$w_{\min}\gets$ smallest positive scaled edge weight}
\State $k\gets1$;\ $\alpha_0\gets0$;\ $E_0,F_0\gets\emptyset$;\
       $\bar\alpha\gets W/w_{\min}+1$
\While{$\textsc{Test}_k(\bar\alpha)$ fails}
  \State $\alpha_k\gets\max\{\alpha\in[\alpha_{k-1},\bar\alpha]:
         \textsc{Test}_k(\alpha)\}$
         \Comment{Section~\ref{sec:exact}}
  \For{every pair or singleton $A$}
    \State compute the minimum and maximum of $x(A)$ over $P_k=R_k(\alpha_k)$
  \EndFor
  \State $E_k,F_k,c\gets$ the pairs and players whose extremes agree
  \State $k\gets k+1$
\EndWhile
\State \rev{find $z\in D_{k-1}$ using Lemma~\ref{lem:pair}}
\State \Return \rev{$\lambda^{-1}D_{k-1}$ and $z/\lambda$}
\end{algorithmic}
\end{algorithm}

\begin{proof}[Proof of Theorem~\ref{thm:algorithm}]
We show that Algorithm~\ref{alg:nucleon} is correct and runs in
polynomial time.

\step{Integral weights.}
Multiply all weights by the least common multiple $\lambda$ of their
denominators. This multiplies $W$ and all allocations by $\lambda$ and
leaves every ratio, and $\bar\alpha$, unchanged. We run the algorithm on
the scaled weights, so that Lemma~\ref{lem:precision} applies, and divide
the output by $\lambda$.

\step{Trivial cases.}
If $W=0$, nonnegativity and efficiency force the zero allocation. If
$n\leq2$, no proper coalition has positive value, so the nucleon is $I$.
Otherwise $n\geq3$ and $W>0$, and no pair or singleton is fixed on
$P_0=I$, so $E_0=F_0=\emptyset$.

\step{Termination test.}
Consider stage $k$, where $E_{k-1}$, $F_{k-1}$ and their values
$c(\cdot)$ are known exactly. If $U_{k-1}\neq\emptyset$, then
$R_k(\bar\alpha)=\emptyset$, because no coalition in $U_{k-1}$ can reach
ratio $\bar\alpha$. If $U_{k-1}=\emptyset$, then
$R_k(\bar\alpha)=P_{k-1}\neq\emptyset$. Moreover
$\bar\alpha>\alpha_{k-1}$: this is clear for $k=1$, and for $k\geq2$ any
$S\in U_{k-2}$ and $x\in P_{k-1}$ satisfy
$\alpha_{k-1}v(S)\leq x(S)<\bar\alpha v(S)$. By
Proposition~\ref{prop:compact}(b), the test at $\bar\alpha$ therefore
succeeds exactly when $U_{k-1}=\emptyset$, that is, when $k-1=K$. In that
case Proposition~\ref{prop:compact} gives $\nucleon(N,v)=D_{k-1}$, so the
algorithm returns the nucleon in the form
\begin{equation}\label{eq:output}
 \left\{x\in\R^N_{\geq0}:
 \begin{array}{l}
 x(N)=W,\\
 x_i+x_j=c(ij)\quad(ij\in E_{k-1}),\\
 x_i=c(i)\quad(i\in F_{k-1})
 \end{array}\right\},
\end{equation}
together with a member of it computed by Lemma~\ref{lem:pair}.

\step{One stage.}
Otherwise, $\textsc{Test}_k(\alpha_{k-1})$ succeeds by
Lemma~\ref{lem:stages}(a), and $\textsc{Test}_k(\bar\alpha)$ fails.
The binary search of Section~\ref{sec:exact} therefore returns $\alpha_k$
exactly. A pair or singleton is fixed on $P_k$ exactly when its minimum
and maximum over $P_k$ agree, so the algorithm stores precisely
$E_k$, $F_k$ and their values. Proposition~\ref{prop:compact}(a) for
$k+1$ guarantees that these equalities also determine every other fixed
coalition payoff, so the next stage is the accelerated refinement of
Section~\ref{sec:stages}.

\step{Running time.}
By Lemma~\ref{lem:stages}(b) there are at most $n-1$ stages. Each stage
runs $O(n^2)$ binary searches, and each search makes polynomially many
tests by Lemma~\ref{lem:precision}. A test performs $O(n^2)$
maximum-weight matching computations, one Bellman--Ford computation and
two minimum-cost flow computations on a network with $2n$ nodes and
$O(n^2)$ arcs, all on numbers of polynomial encoding length. Thus the
algorithm runs in polynomial time.
\end{proof}

In Example~\ref{ex:triangle}, the first stage fixes all three pair sums at
$2/3$, and the algorithm immediately returns the equal split.
Example~\ref{ex:path} shows a run with three stages.

\section{Hardness for \texorpdfstring{$b$}{b}-matching and partitioned
matching games}\label{sec:hardness}

The algorithm above uses the fact that a matching decomposes a coalition
into disjoint pairs and possibly unmatched players. Capacities allow
several selected edges to share a player, and in a partitioned game one
player can own both ends of several edges, so pair sums no longer provide
that decomposition. We prove hardness through a different connection:
the total payoff assigned to a new independent partnership records the
stability of the original game.

\begin{theorem}\label{thm:hardness}
Computing an exact member of the nucleon of a uniform $b$-matching game
is NP-hard, even when $b_i\leq3$ for every player $i$.
\end{theorem}

\begin{theorem}\label{thm:partitioned}
Computing an exact member of the nucleon of a uniform partitioned
matching game is NP-hard, even when every player owns at most three
vertices.
\end{theorem}

\begin{remark}\label{rem:reductions}
Both theorems concern producing one allocation, although the nucleon may
contain several. The reductions behind them make a single call to an
algorithm that returns a nucleon allocation, and then add the payoffs of
two specified players. They work whichever nucleon allocation the
algorithm returns, and they do not need the optimal first-stage ratio as
an additional output.
\end{remark}

{Both reductions are instances of the following direct-sum lemma.}

\subsection{{Direct sums calibrate component budgets}}

{The hardness gadget is most transparent as a general fact about direct
sums. Let $(N_1,v_1),\ldots,(N_m,v_m)$, $m\geq2$, be nonnegative games on
pairwise disjoint player sets, with $W_j=v_j(N_j)>0$. Their \emph{direct
sum} is the game {$(\bar N,\bar v)$ on $\bar N=\dot\bigcup_{j=1}^m N_j$}
defined by
\begin{equation}\label{eq:direct-sum}
 {\bar v(S)}=\sum_{j=1}^m v_j(S\cap N_j).
\end{equation}
For component $j$, define its best multiplicative guarantee when its own
grand coalition is also compared,
\begin{equation}\label{eq:component-ratio}
 r_j=\max_{y\geq0:\,y(N_j)=W_j}
       \ \min_{T\subseteq N_j:\,v_j(T)>0}\frac{y(T)}{v_j(T)}.
\end{equation}
Because $N_j$ itself is included, $0<r_j\leq1$, and $r_j=1$ if and only if
$\core(N_j,v_j)\neq\emptyset$.

\begin{lemma}[Direct-sum calibration]\label{lem:calibration}
Let {$\bar W=\sum_{j=1}^m W_j$}. The first-stage optimum of the direct-sum game is
\begin{equation}\label{eq:direct-sum-alpha}
 \alpha^*=\frac{{\bar W}}{\displaystyle\sum_{j=1}^m W_j/r_j},
\end{equation}
and every first-stage optimal imputation assigns the uniquely determined
component totals
\begin{equation}\label{eq:direct-sum-budget}
 x(N_j)=\alpha^*\frac{W_j}{r_j}\qquad(j=1,\ldots,m).
\end{equation}
Consequently, $\alpha^*\leq1$, with equality if and only if every component
has a nonempty core.
\end{lemma}

\begin{proof}
Fix an imputation $x$ of the direct sum and write $s_j=x(N_j)$. Within
component $j$, scaling an optimizer of \eqref{eq:component-ratio} shows
that the largest achievable minimum ratio with total budget $s_j$ is
$r_js_j/W_j$. Moreover, the minimum ratio over all positive-valued proper
coalitions of the direct sum equals the minimum, over all $j$, of the
{smallest ratio $x(T)/v_j(T)$ over coalitions $T\subseteq N_j$ with
$v_j(T)>0$}. Indeed, each component coalition is proper in the direct sum, while
the ratio of any coalition meeting several components is a weighted average
of the ratios of its positive-valued component parts (zero-valued parts can
only add nonnegative payoff).

Hence the first stage reduces to
\[
 \max\left\{\min_j\frac{r_js_j}{W_j}:s_j\geq0,\ 
                  \sum_j s_j={\bar W}\right\}.
\]
If the common guaranteed ratio is $\alpha$, then necessarily
$s_j\geq \alpha W_j/r_j$ for every $j$. Summing gives
$\alpha\leq {\bar W}/(\sum_j W_j/r_j)$, and equality is attained by setting all
these inequalities to equality. This yields
\eqref{eq:direct-sum-alpha}--\eqref{eq:direct-sum-budget} and also proves
uniqueness of the component totals. Since every $r_j\leq1$, equality
$\alpha^*=1$ holds exactly when $r_j=1$ for all $j$, equivalently when all
component cores are nonempty.
\end{proof}

The reduction below uses the special case in which one component is the
original game $(N,v)$, with $W=v(N)>0$, and the other is a two-player
unit-edge game on $\{a,a'\}$. The latter has $W_2=r_2=1$, and the direct
sum has value function
\begin{equation}\label{eq:extension}
 v'(S)=v(S\cap N)+
 \begin{cases}
 1,&\{a,a'\}\subseteq S,\\
 0,&\text{otherwise}.
 \end{cases}
\end{equation}
By Lemma~\ref{lem:calibration}, every first-stage optimum---hence every
nucleon allocation---has
\begin{equation}\label{eq:core-test}
 x_a+x_{a'}\leq1,
 \qquad
 x_a+x_{a'}=1\ \Longleftrightarrow\ \core(N,v)\neq\emptyset.
\end{equation}
Thus the isolated edge is a one-number certificate of whether the original
component can satisfy all of its coalitions simultaneously. For the unit
triangle of Example~\ref{ex:triangle}, $r_1=2/3$ and $W_1=1$, giving
$x_a+x_{a'}=4/5$ as in Figure~\ref{fig:triangle}(b).}

\subsection{The reductions}

We now apply the lemma to the hard core instances. In both cases, all
that changes in the graph is the addition of one isolated edge.

\begin{proof}[Proof of Theorem~\ref{thm:hardness}]
By \citet[Theorem~1]{BBK+25a}, deciding core nonemptiness is
coNP-hard for uniform $b$-matching games with capacities at most three.
Its complement, deciding core emptiness, is therefore NP-hard under
polynomial-time many-one reductions. Instances with $W=0$ have the zero
allocation in their core and can be decided directly, so we may assume
$W>0$.

Given such an instance $(G,b)$, add a disjoint unit-weight edge $aa'$
and give $a,a'$ capacity one. The enlarged game has exactly the value
function \eqref{eq:extension}, because $b$-matchings in the two graph
components can be chosen independently. The construction preserves
unit weights and the upper bound three on capacities.

Suppose an algorithm returns any exact nucleon allocation $x$ of this
enlarged game, and let $t=x_a+x_{a'}$. Lemma~\ref{lem:calibration}
gives the decision rule
\[
 \core(N,v)=\emptyset\quad\Longleftrightarrow\quad t<1.
\]
Thus one allocation call solves an NP-hard decision problem. The
argument applies to every possible nucleon output, so no choice among
allocations avoids the reduction.
\end{proof}

\begin{proof}[Proof of Theorem~\ref{thm:partitioned}]
By \citet[Corollary~1]{BBK+25a}, deciding core nonemptiness is coNP-hard
for uniform partitioned matching games of width at most three; they
obtain this from the $b$-matching result above through a reduction that
preserves unit weights and turns capacity three into width three. Hence
deciding core emptiness is NP-hard for these games. As before, instances
with $W=0$ can be decided directly.

Given such an instance, add two vertices $a$ and $a'$ joined by a
unit-weight edge, and two new players, also called $a$ and $a'$, who own
$\{a\}$ and $\{a'\}$ respectively. The width stays at most three and all
weights remain one. A matching in the subgraph owned by a coalition $S$
consists of a matching in the original graph, owned by $S\cap N$, and
possibly the edge $aa'$, which is available exactly when both new players
belong to $S$. Hence the enlarged game has the value function
\eqref{eq:extension}.

{The specialization of Lemma~\ref{lem:calibration} only requires the
original component to be nonnegative with $W>0$; in particular, it allows
singletons of positive value, which occur in partitioned matching games.
Therefore any exact nucleon allocation $x$ of the enlarged game again
satisfies $\core(N,v)=\emptyset$ if and only if $x_a+x_{a'}<1$.}
\end{proof}

The hardness already comes from the multiplicative least-core stage:
every allocation that attains its optimum in the enlarged instance
has the edge total in Lemma~\ref{lem:calibration}. Later lexicographic
comparisons cannot erase the decision encoded by that total, so finding
a first-stage optimal allocation is NP-hard as well, for both classes.
{The optimal value is also hard to compute: for the original component,
the parameter $r_1$ in \eqref{eq:component-ratio} equals
$\min\{1,\alpha_1(N,v)\}$, because the only additional comparison is the
grand coalition $N$, whose ratio is one. Thus deciding whether
$\alpha_1(N,v)\geq1$ is coNP-hard for both classes by the same results of
\citet{BBK+25a}. The disjoint edge transfers this information from the
optimal value to every optimal allocation.} The conclusions concern exact computation and do not
establish hardness of approximation.

\Needspace{16\baselineskip}
\section{The capacity-two frontier}\label{sec:capacity-two}

{The complexity of exact nucleon computation for $b$-matching games
with $b_i\leq2$ remains open. The first multiplicative least-core stage is
polynomially separable using the known capacity-two core machinery
\citep{SV25a}. \rev{To exclude the grand coalition, check $x\geq0$,
$x(N)=W$ and $\alpha\geq0$ separately, and apply the path and cycle
inequality-separation routines of \citet{SV25a} to each induced graph
$G-r$, with weights $\alpha w$ and the restricted payoff vector. No
efficiency equality is imposed on an induced game. These routines
separate the coalition inequalities because every $b$-matching decomposes
into vertex-disjoint paths and cycles, and uncovered vertices have
nonnegative payoff. Every proper coalition omits some vertex $r$, so the
$n$ calls detect exactly the possible violations. This remains valid when
the optimal ratio exceeds one.}

The difficulty in later stages is different: coalition
constraints whose payoffs are already fixed must be excluded. The recent
linear-subspace-avoidance framework of \citet{EE26a} isolates exactly this
obstacle for the nucleolus. Here we only need to observe that the same
oracle also separates the multiplicative stages of the nucleon.

For a nonempty polytope $P\subseteq\R^N$, let
\[
 L(P)=\lin(P-P)^\perp.
\]
Identifying a coalition $S$ with its incidence vector $\chi^S$, the payoff
$x(S)=\langle\chi^S,x\rangle$ is constant on $P$ if and only if
$\chi^S\in L(P)$. Thus, with $L_{k-1}=L(P_{k-1})$, \rev{the
positive-valued proper coalitions still relevant at stage $k$ are exactly
those whose incidence vectors lie outside $L_{k-1}$.}

\begin{definition}[{\LSA{} {\citep{EE26a}}}]
Given a cooperative game $(N,v)$, a rational vector $y\in\mathbb Q^N$,
and a proper rational linear subspace $L\subsetneq\R^N$ specified by a
basis, \LSA{} asks for a coalition $S\subseteq N$ with $\chi^S\notin L$
minimizing the excess $y(S)-v(S)$.
\end{definition}

\begin{definition}[{\SNZC{} {\citep{EE26a}}}]
Given an undirected graph $H=(V,E)$, edge costs $c\in\mathbb Q^E$ that are
\emph{conservative} (every cycle has nonnegative total cost), and integral
edge labels $a\in\mathbb Z^E$, \SNZC{} asks for a minimum-cost cycle $C$
with $a(C)=\sum_{e\in C}a_e\neq0$.
\end{definition}

\citet[Theorem~3.1]{EE26a} prove that, for simple $b$-matching games with
$b\leq2$, \LSA{} and \SNZC{} are polynomially equivalent{, via an
intermediate problem of finding a maximum-weight matching whose label sum
is nonzero}. We use this result as a black box; the only new
point is \rev{how separation at a multiplicative nucleon stage reduces
to an \LSA{} instance.}
For a candidate $x\in P_{k-1}$ and $\alpha\geq0$, separation asks whether
some unfixed positive-valued proper coalition violates
$x(S)\geq\alpha v_w(S)$, equivalently whether
\begin{equation}\label{eq:b2-lsa}
 \min_{\chi^S\notin L_{k-1}}
       \{x(S)-\alpha v_w(S)\}<0.
\end{equation}
Since scaling nonnegative edge weights commutes with $b$-matching
optimization,
\[
 v_{\alpha w}(S)=\alpha v_w(S)\rev{.}
\]
\rev{Thus the minimum in \eqref{eq:b2-lsa} can be computed by \LSA{} for
the capacity-two game with weights $\alpha w$ and vector $x$, and its sign
decides whether a violation exists.} No exceptional coalitions are hidden in
this reformulation: $\chi^\emptyset=0\in L_{k-1}$; efficiency makes
$\chi^N\in L_{k-1}$; and a zero-valued coalition has
$x(S)-\alpha v_w(S)=x(S)\geq0$, so it cannot witness a violation.

\begin{theorem}\label{thm:b2-snzc}
Computing an exact nucleon allocation of a $b$-matching game with
nonnegative rational edge weights and $b_i\leq2$ for every player $i$ is
polynomial-time Turing reducible to \SNZC{}.
\end{theorem}

\begin{proof}
\rev{We implement the successive stages using the preceding separation
routine. If $v_w(N)=0$, return the zero allocation. Otherwise multiply
the weights by the least common multiple $\lambda$ of their denominators
and work with this integer weighting, still denoted by $w$. At the end
we divide the resulting allocation by $\lambda$. Now $W=v_w(N)>0$,
$P_0=I$ is given explicitly, and $L_0=L(I)$ is spanned by $\chi^N$.

Before each stage, test terminality. If $L_{k-1}=\R^N$, the stage is
terminal. Otherwise call \LSA{} with $y=0$, subspace $L_{k-1}$, and
weights $w$, without multiplying them by a candidate $\alpha$. Its
optimum is negative exactly when a positive-valued coalition lies outside
$L_{k-1}$. Since $\emptyset$ and $N$ lie in $L_{k-1}$, this is precisely
the condition that the stage is nonterminal.

At a nonterminal stage, impose $0\leq\alpha\leq W$. This loses no optimum:
every relevant coalition has integer value at least one and payoff at
most $W$. For a candidate $(x,\alpha)$, first check these bounds and
separate $x$ from $P_{k-1}$ using the preceding stage's oracle. If
$x\in P_{k-1}$, the observation above and
\citet[Theorem~3.1]{EE26a} separate the new coalition constraints using
polynomially many \SNZC{} calls. Thus we obtain a strong separation
oracle for the bounded stage polyhedron. We optimize the current level,
recover the affine hull of the optimal polytope to determine
$L_k$, and continue on that polytope, following the successive-stage
framework of \citet{MPS79a} in its multiplicative form
(Section~\ref{sec:stages}) and its oracle implementation by
\citet{EE26a}. Lemma~\ref{lem:stages}(b) bounds the number of nonterminal
stages by $n-1$.

The reduction is exact over rational inputs. Lemma~\ref{lem:precision}
supplies polynomial encoding bounds for all stage levels and vertices:
its determinant argument uses only integrality of coalition values,
the bound $v(S)<\Omega=1+\sum_e w(e)$, and Lemma~\ref{lem:stages}, so it
applies verbatim to $b$-matching games. The standard exact
separation--optimization machinery \citep{GLS88a} therefore performs
the optimizations, affine-hull computations and equality tests in
polynomial time. By Lemma~\ref{lem:stages}(c), the terminal polytope is
the nucleon of the scaled game. A vertex of it, found by the same
machinery and divided by $\lambda$, is an exact nucleon allocation of
the original game.}
\end{proof}

\rev{The qualification in Theorem~\ref{thm:b2-snzc} is important. The
\emph{unrestricted} \LSA{} optimization oracle used by our separation
routine is polynomially equivalent to \SNZC{} by \citet{EE26a}. Actual
stage queries have additional structure: the payoff vector is
nonnegative and efficient, the avoided subspace comes from a preceding
optimal polytope, and separation only needs a violated inequality or a
certificate that none exists. We establish no reverse reduction for
these restricted queries or for full nucleon computation. Thus a
polynomial-time algorithm for \SNZC{} would imply one for capacity-two
nucleon computation, while the unconditional complexity of the latter
remains open.}}

\Needspace{12\baselineskip}
\section{Conclusions}

The nucleon of a matching game admits an exact polynomial-time
combinatorial algorithm, whereas finding any nucleon allocation is
NP-hard for unit-weight $b$-matching games with capacities at most three
and for unit-weight partitioned matching games of width at most three.
These results resolve Open Problems~5 and~7 of \citet{BBJ+23a}. Two
questions remain. Can the matching-game algorithm be made strongly
polynomial? What is the complexity of nucleon computation when every
capacity is at most two, or when every player owns at most two vertices?
In both cases core nonemptiness can be decided in polynomial time
\citep{BBK+25a}, so our first-stage reduction does not apply. For
capacities at most two, \rev{Section~\ref{sec:capacity-two} gives a
polynomial-time Turing reduction to \SNZC{} through an unrestricted
minimum-excess optimization oracle. Whether the restricted separation
queries arising in actual stages admit a simpler algorithm remains open.
The width-two partitioned case also remains open, without a comparable
oracle reduction.}

\section*{Declaration of generative AI assistance}
The authors acknowledge the assistance of OpenAI's GPT Sol 6.1 and Anthropic's Fable 5.1 in working out the proof details of the combinatorial algorithm.

% \section*{Acknowledgments}

% Haris Aziz is supported by the NSF-CSIRO grant on ``Fair Sequential
% Collective Decision-Making'' (RG230833).

%\bibliographystyle{elsarticle-num-names}
\bibliographystyle{plainnat}
% Audit correction: use the basenames of the supplied bibliography files.
\bibliography{references}

\end{document}